\documentclass[11pt]{article}
\usepackage[utf8]{inputenc}
\usepackage[T1]{fontenc}

\usepackage{amsmath, amssymb, amsthm, amsfonts}
\usepackage{mathtools}
\usepackage{thm-restate}

\usepackage[dvipsnames,svgnames,x11names]{xcolor}
\definecolor{ForestGreen}{rgb}{0.1333,0.5451,0.1333}
\usepackage[colorlinks=true,linkcolor=ForestGreen,citecolor=blue]{hyperref}

\usepackage[margin=1in]{geometry}
\usepackage{graphics}
\usepackage{pifont}
\usepackage{tikz}
\usepackage{bbm}
\usepackage[T1]{fontenc}

\usetikzlibrary{arrows.meta}
\usepackage{environ}
\usepackage{framed}
\usepackage{url}
\usepackage{algorithm}
\usepackage[algo2e,linesnumbered,lined,ruled,boxed]{algorithm2e}
\usepackage{algpseudocode}
\usepackage[noabbrev,capitalize,nameinlink]{cleveref}
\crefname{equation}{}{}
\usepackage[labelfont=bf]{caption}
\usepackage{cite}
\usepackage{framed}
\usepackage[framemethod=tikz]{mdframed}
\usepackage{appendix}
\usepackage{graphicx}
\usepackage[textsize=tiny]{todonotes}
\usepackage{tcolorbox}
\allowdisplaybreaks[1]

\usepackage{fancybox}

\newenvironment{fminipage}%
  {\begin{Sbox}\begin{minipage}}%
  {\end{minipage}\end{Sbox}\fbox{\TheSbox}}

\newenvironment{algbox}[0]{\vskip 0.2in
\noindent 
\begin{fminipage}{6.3in}
}{
\end{fminipage}
\vskip 0.2in
}

\newcommand\remove[1]{}

\newtheorem{lemma}{Lemma}[section]
\newtheorem{theorem}{Theorem}
\newtheorem*{lemma*}{Lemma}

\newtheorem*{corollary*}{Corollary}

\theoremstyle{definition}

\newtheorem*{theorem*}{Theorem}
\newtheorem{definition}[lemma]{Definition}

\newtheorem*{rem*}{Remark}

\newtheorem{fact}[lemma]{Fact}

\newcommand{\R}{\mathbb{R}}

\newcommand{\E}{\mathop{\mathbb{E}}}
\renewcommand{\O}{\widetilde{O}}
\newcommand{\one}{\mathbbm{1}}

\crefname{algocf}{Algorithm}{Algorithms}

\newcommand{\poly}{\mathrm{poly}}

\renewcommand{\bar}{\overline}
\renewcommand{\hat}{\widehat}

\renewcommand{\bar}{\overline}

\newcommand{\Tr}{\mathsf{Tr}}

\newcommand{\F}{\mathbb{F}}

\renewcommand{\ge}{\geqslant}
\renewcommand{\le}{\leqslant}
\renewcommand{\geq}{\geqslant}
\renewcommand{\leq}{\leqslant}

\newcommand\Otil{\widetilde{O}}

\renewcommand\AA{\boldsymbol{\mathit{A}}}
\newcommand\BB{\boldsymbol{\mathit{B}}}
\newcommand\CC{\boldsymbol{\mathit{C}}}
\newcommand\DD{\boldsymbol{\mathit{D}}}
\newcommand\II{\boldsymbol{\mathit{I}}}
\newcommand\LL{\boldsymbol{\mathit{L}}}
\newcommand\MM{\boldsymbol{\mathit{M}}}
\newcommand\WW{\boldsymbol{\mathit{W}}}
\newcommand\bb{\boldsymbol{\mathit{b}}}
\newcommand\rr{\boldsymbol{\mathit{r}}}
\newcommand\ww{\boldsymbol{\mathit{w}}}
\newcommand\xx{\boldsymbol{\mathit{x}}}
\newcommand{\cc}{\boldsymbol{\mathit{c}}}
\newcommand{\hh}{\boldsymbol{\mathit{h}}}

\newcommand\xxtil{\widetilde{\boldsymbol{\mathit{x}}}}

\newcommand\ttau{\boldsymbol{\tau}}
\newcommand\ttautil{\widetilde{\boldsymbol{\tau}}}
\newcommand\ttheta{\boldsymbol{\theta}}

\newcommand\uu{\boldsymbol{\mathit{u}}}
\newcommand\vv{\boldsymbol{\mathit{v}}}

\newcommand{\eps}{\varepsilon}
\newcommand{\diag}{\mathsf{diag}}
\newcommand{\tr}{\mathsf{Tr}}
\renewcommand{\Otil}{\widetilde{O}}
\newcommand{\cT}{\mathcal{T}}
\renewcommand{\det}{\mathsf{det}}

\newcommand{\calD}{\mathcal{D}}
\newcommand{\Var}{\mathrm{Var}}

\title{Approximate Spanning Tree Counting
from \\ Uncorrelated Edge Sets}

\author{Yang P. Liu\\
CMU\\
\texttt{yangl7@andrew.cmu.edu}
\and
Richard Peng\\
CMU\\
\texttt{yangp@cs.cmu.edu}
\and
Junzhao Yang\\
CMU\\
\texttt{junzhaoy@andrew.cmu.edu}}

\begin{document}

\maketitle

\begin{abstract}
We show an $\widetilde{O}(m^{1.5} \epsilon^{-1})$\footnote{We use $\widetilde{O}(\cdot)$ to suppress polylogarithmic factors in $m$.} time algorithm that on a graph with $m$ edges and $n$ vertices outputs its spanning tree count up to a multiplicative $(1+\epsilon)$ factor with high probability,
improving on the previous best runtime of $\widetilde{O}(m + n^{1.875}\epsilon^{-7/4})$ in sparse graphs.
While previous algorithms were based on computing Schur complements and determinantal sparsifiers, our algorithm instead repeatedly removes sets of uncorrelated edges found using the electrical flow localization theorem of Schild-Rao-Srivastava [SODA 2018].
\end{abstract}

\section{Introduction}
\label{sec:intro}

Kirchoff's matrix tree theorem is one of the most important connections
between algebraic and combinatorial graph theory.
It states that in undirected weighted graphs, the total spanning
tree weight is the determinant of a minor of the graph Laplacian matrix.
Laplacian matrices~\cite{V13:book},
determinants~\cite{KT12:book},
spanning tree distributions~\cite{B89},
and related quantities~\cite{L93:survey}
are all topics central to combinatorial optimization,
graph theory, and continuous algorithms.

Algorithmic studies revolving around the matrix tree theorem,
both in determinant estimation and spanning tree sampling,
have motivated the development of a multitude of tools.
They include graph partitioning~\cite{KM09},
structural properties of effective resistances~\cite{MST15,DKPRS17,S18},
matrix concentration~\cite{DPPR20:journal},
and Markov chain mixing~\cite{ALGV19,CGM19,ALGVV21,ALV22}.
While the last of these led to near-optimal algorithms
for sampling spanning trees, the current fastest algorithms
for estimating spanning tree counts up to $(1+\epsilon)$-multiplicative approximations are still only just
below quadratic time in sparse graphs, at about $m + n^{15/8}$ \cite{DPPR20:journal,CGPSSW23:journal}.

The main difficulty of estimating the Laplacian determinant/spanning tree
count is that the value to be computed is quite large (intuitively, the product of $n$ eigenvalues).
Put another way, the logarithm of the spanning tree count is about $n$, and we wish to estimate this quantity up to an additive $\epsilon$. In general, algorithms for estimating such quantities have runtimes depending polynomially on the relative accuracy, which is about $\epsilon/n$, which would lead to large runtimes.
Previous algorithms do repeated vertex elimination \cite{DPPR20:journal,CGPSSW23:journal}, and sparsify the resulting dense graph.
The number of edges in the sparsified graph is important in determining the runtime. In the even more general setting of rank-$n$ matroids in a universe of size $m$, the works \cite{AD20,ADVY22} proved that one can sparsify down to a universe of size about $O(n^2)$ while approximately preserving the number of bases. For spanning trees, \cite{DPPR20:journal} proved that one can sparsify down to an even smaller size of $\O(n^{1.5})$.
Our algorithm does not use a Schur complement and sparsify approach, and instead repeatedly deletes subsets of ``uncorrelated edges'' obtained
via the localization of electrical flows~\cite{SRS18}.
Our main result is an improved running time for estimating
spanning tree counts in sparse graphs.

\begin{theorem}
\label{thm:main}
There is a routine $\textsc{ApproxSpanningTree}(G, \eps)$
that takes as input an undirected graph $G$
with $n$ vertices, $m$ edges, and polynomially bounded edge weights,
along with an error threshold $0 < \eps < 1$,
and outputs in $\Otil(m^{1.5} \eps^{-1})$ time a $(1+\eps)$ multiplicative approximation of the spanning tree count of the graph $G$ with high probability.
\end{theorem}

This algorithm represents a different approach to Laplacian
based algorithms.
Even with a likely optimal effective resistance estimation
routine that produces $\epsilon$-approximate resistances
in $m/\epsilon$ time (the current best known runtime is about $m/\epsilon^{1.5}$ in general graphs \cite{CGPSSW23:journal}),
previous algorithms only get a running time of $m + n^{1.75}$ (slower than our algorithm in sparse graphs).
We believe ideas from our approach may be useful in improving other Laplcian based 
algorithms that utilize sparsification and elimination~\cite{BLNPSSSW20,BLLSSSW21,GLP21,BGJLLPS22}.

\section{Preliminaries}
\label{sec:prelims}

The graphs in this paper are undirected and connected,
with positive weights on edges.
We denote them as triples of vertices, edges, and weights:
$G = (V, E, \ww)$.
We use $n = |V|$ to denote the number of vertices,
$m = |E|$ to denote the number of edges.

The weight of a spanning tree $T$ is
the product of the weights of all edges in $T$.
The sum over all the trees of their weights is $\cT(G)$,
the total spanning tree weight of $G$.

The graph Laplacian matrix $\LL(G)$ is an $n$-by-$n$ matrix
given by putting the weighted degrees on the diagonals,
and the negation of edge weights on off diagonals.
It can alternately be defined as
a sum of rank-$1$ matrices, one per edge.
For an edge $e = uv$, we let its indicator vector $\bb(e)$
be the vector with $1$ in entry $u$, $-1$ in entry $v$:
this orientation can be arbitrary.
Then the Laplacian matrix is given by
\[
\LL\left( G \right)
=
\sum_{e} \ww_e  \bb(e) \bb(e)^{\top}
=
\BB^{\top} \WW \BB,
\]
where $\WW$ is a $m$-by-$m$ diagonal matrix with all the edge
weights, and $\BB \in \R^{E \times V}$ is the $m \times n$ edge-vertex
incidence matrix whose rows are the vectors $\bb(e)$.
When context is clear,
we omit the $(G)$ and just write $\LL$ to denote the Laplacian matrix.

For a connected graph $G$, the null space of $\LL(G)$
is precisely the span of the all-$1$s vector.
The pseudoinverse of $\LL(G)$, which inverts on the
space orthogonal to this vector, and is $0$ on this vector,
is given by $\LL(G)^{\dag}$.
We also define the non-zero determinant, $\det^+(\MM)$.
For a symmetric matrix $\MM$, its determinant is the product
of all the non-zero eigenvalues of $\MM$.

The matrix tree theorem states that
$\cT(G) = \frac1n\det^{+}(\LL(G))$.
The fraction of spanning trees that an edge $e$ is
involved in can also be calculated via the solution of
a system of linear equations.
Specifically, 
for any edge $e$ of $G$, the spanning tree weight of $G$ and $G \setminus \{e\}$ ($G$ with $e$ removed) are related by
\[
1 - \frac{\cT\left( G \setminus \left\{ e\right\}\right)}
{\cT\left( G \right)}
=
\ww_e \bb\left( e \right)^{\top} \LL\left( G \right)^{\dag} \bb\left( e \right).
\]
We denote this quantity as $\ttau_e$. In the literature, this is referred to as the \emph{leverage score} of the edge $e$, and is widely used throughout numerical linear algebra.

We will use a more sophisticated version of this
formula involving multiple edges.
However, we still need to solve systems of equations in $\LL$.
For this we need to use nearly-linear time Laplacian
solvers~\cite{ST14:journal}.

\begin{lemma}
\label{lem:laplacian_solver}
There is a routine $\textsc{Solve}(\LL, \bb)$ that
given a graph $G$ with polynomial bounded weights
and a vector $\bb$ orthogonal to $\vec{1}$,
returns in $\Otil(m)$ time a vector $\xxtil$  satisfying $\|\xxtil - \LL^{-1}\bb\|_2 \le 1/\poly(n)$.
\end{lemma}

We keep all entries polynomial bounded, and work with polynomially small
error to simplify away issues involving bit complexity and
representations of numbers.

For a random variable $X$ we let $\E[X]$ denote its expectation and $\Var[X]$ denote its variance.
\section{Overview}
\label{sec:algorithm}

Our algorithm is based on repeatedly removing a subset of edges
$F \subseteq E(G)$ and relating the spanning tree counts of the two graphs
using determinant expansion
\begin{multline*}
\det^+\left( \LL\left( G \right)  - \LL\left( F \right) \right)
=
\det^+\left( \LL\left( G \right)  - \BB\left( F \right)^{\top} \WW\left( F \right) \BB\left( F \right)\right)\\
=
\det^+\left( \LL\left( G \right)\right)
\cdot
\det\left( \II - \LL\left( G \right)^{\dag 1/2}
\BB\left( F \right)^{\top} \WW\left( F \right) \BB\left( F \right)
\LL\left( G \right)^{\dag 1/2}\right).
\end{multline*}
The latter term can be turned into the determinant of an $|F|$-by-$|F|$
matrix via the fact
\begin{fact}
For any (possibly asymmetric) matrix $\AA$,
$\det(\II - \AA \AA^{\top}) = \det(\II - \AA^{\top} \AA)$.
\end{fact}
\begin{proof}
Follows from the fact that $\AA^\top \AA$ and $\AA \AA^\top$ have the same nonzero eigenvalues.
\end{proof}

Throughout, it is more convenient to track logarithmic error.
That is, we instead approximate
\begin{equation}
\log \det\left( \II -
\WW\left( F \right)^{1/2}\BB\left( F \right)
\LL\left( G \right)^{\dag}
\BB\left( F \right)^\top \WW\left( F \right)^{1/2} \right). \label{eq:m}
\end{equation}
We will try to build an approximator to the above expression that has low additive error and low variance.
To this end, define the $|F| \times |F|$ matrix $\MM =  \II -
\WW\left( F \right)^{1/2}\BB\left( F \right)
\LL\left( G \right)^{\dag}
\BB\left( F \right)^\top \WW\left( F \right)^{1/2}$, so that our goal is to estimate $\log \det (\MM)$. The entries of $\MM$ are:
\[
\MM_{e, f}
=
\begin{cases}
1 - \ttau_e
& \text{if $e = f$}\\
-\ww_{e}^{1/2} \bb\left( e \right)^{\top}
\LL^{\dag}
\bb\left( f \right) \ww_f^{1/2}
& \text{if $e \neq f$}
\end{cases}
\]
The first key observation is that when the off-diagonals are sufficiently
small, $\det(\MM)$ is approximately the product of its diagonal terms.

\begin{lemma}
\label{lem:taylor}
If $\MM$ is a $k$-by-$k$ matrix such that $\MM_{ii} \ge 0.1$ for all $i \in [k]$ and $\sum_{j \neq i} |\MM_{ij}| \le \rho$ for all $i \in [k]$ for some $\rho \le 0.01$, then
\[
\left|\log \det\left(M\right)
- \sum_{i = 1}^{k} \log \left(\MM_{ii} \right)\right|
\le
O\left(\rho^2  k \right).
\]
\end{lemma}

\begin{proof}
Let $\DD$ be the diagonal matrix consisting of the diagonal
of $\MM$, and $\AA = \MM - \DD$ be the off-diagonal entries.
Note that it is important that the diagonal entries of $\AA$ are $0$.

The log determinant of $\MM$, up to a scaling by $\DD$
can be written as
\begin{align*}
\log\det\left(\MM\right)
&=
\log\det\left(\DD + \AA\right)
=
\log\det\left(\DD\right) + \log\det\left(\II + \DD^{-1/2}\AA\DD^{-1/2}\right)  \\
&=
\left( \sum_{i = 1}^{k} \log\left(\DD_{ii}\right) \right)
+ \log\det\left(\II + \DD^{-1/2}\AA \DD^{-1/2}\right) .
\end{align*}
The conclusion of the lemma is equivalent to bounding the magnitude of the second term.
For a symmetric matrix, the spectral theorem gives
that the log determinant is the trace of the log.
So symmetrizing and applying Taylor expansion to
$\log(\II + \DD^{-1/2} \AA \DD^{-1/2})$ gives
\[
\tr\left[ \log\left( \II + \DD^{-1/2} \AA \DD^{-1/2} \right) \right]
=
\sum_{t \ge 1} \frac{(-1)^{t-1}}{t} \tr\left[ \left(  \DD^{-1/2} \AA \DD^{-1/2} \right)^t \right]
\]
The $t = 1$ term vanishes because $\DD$ is diagonal and
$\AA$ has $0$ on the diagonal.

To bound the $t \ge 2$ terms, note that
$\|\AA\|_2 \leq \rho$ and $\|\DD^{-1/2}\|^2_2 \leq 10$.
Thus
$\Tr[(\DD^{-1/2} \AA \DD^{-1/2})^t]$ is at most $(10\rho)^tk$,
and the bound follows from summing over $t \ge 2$.
\end{proof}

\cref{lem:taylor} motivates the definition of an \emph{uncorrelated edge set}, which is critical to the remainder of the work. Essentially, a $\rho$-correlated edge set is one where $\sum_{j \neq i} |\MM_{ij}| \le \rho$. In this language, \cref{lem:taylor} is saying that if a set of edges $F$ is $\rho$-uncorrelated, then the events that each edge $e$ is in a spanning tree are very nearly independent.

\begin{definition}
\label{def:uncorrelated_edge_subset}
In a graph $G$, a subset of edges $F \subseteq E$
is $\rho$-correlated if for all $f \in F$ we have
\[
\sum_{e \in F, e \neq f}
\left| \ww_e^{1/2}
\bb\left( e\right)^\top \LL^{\dag} \bb\left( f \right) \ww_{f}^{1/2} \right|
\le
\rho.
\]
\end{definition}

If the set $F$ of edges is sufficiently uncorrelated, it suffices to estimate
\[ \sum_{e \in F} \log \MM_{ee} = \sum_{e \in F} \log(1-\ttau_e). \] To do this, we first use again that $F$ is uncorrelated to produce estimates $\ttautil_e \approx \ttau_e$, do another Taylor expansion, and estimate the first order term (again, the second order and higher terms have negligible contribution).
Before doing this formally, we first develop some useful properties of uncorrelated edges in \cref{sec:uncorrelated_edge_subsets} before giving the overall algorithm in \cref{sec:ball}.

\section{Finding and Using Uncorrelated Edge Subsets}
\label{sec:uncorrelated_edge_subsets}

Uncorrelated edge subsets stem from the electrical flow
localization result by Schild-Rao-Srivastava.
They were first initially used in random spanning tree
sampling~\cite{S18}, but have then been used to construct
Schur complements of graphs which are minors~\cite{LS18}, which has led to near-optimal distributed Laplacian system solvers~\cite{FGLPSY21}.

We first give algorithms that efficiently estimate the leverage scores of all edges within an uncorrelated edge subset, and also give low-variance estimates for linear combinations of leverage scores and their higher moments. This is a key idea of this work.
Afterwards, we describe a standard algorithm for identifying large uncorrelated edge subsets
in \cref{subsec:find_uncorrelated_subsets}. These algorithms are already in the works \cite{LS18,S18}, but we reproduce their short arguments for completeness.

\subsection{Estimating Leverage Scores}

The most direct use of uncorrelated edge subsets is that
we can find each edge's leverage score from the same Laplacian solve,
to a good additive error.

\begin{lemma}
\label{lem:estimate_first}
Given $G = (V, E, \ww)$ with polynomially bounded edge weights
and a $\rho$-correlated edge subset $F$ for $\rho > \frac{1}{\poly(n)}$,
there is an algorithm that in $\O(m)$ time produces estimates $\ttautil_f$ for all $f \in F$ satisfying $|\ttautil_f - \ttau_f| \le 2\rho$ with high probability.
\end{lemma}

\begin{proof}
For each $f \in F$, return the estimate
\[
\ttautil_f := \ww_f^{1/2} \bb\left( f \right)^{\top} \LL^{\dag}
\sum_{e} \bb\left( e \right) \ww_{e}^{1/2}. \]
By the triangle inequality, the error 
is at most
\begin{align*}
|\ttautil_f - \ttau_f| &= \left|\ww_f \bb(f)^\top \LL^{\dag} \bb(f) - \ww_f^{1/2} \bb\left( f \right)^{\top} \LL^{\dag}
\sum_{e} \bb\left( e \right) \ww_{e}^{1/2} \right| \\
&\le \sum_{e \in F, e \neq f} \left|\ww_e^{1/2} \bb(e)^\top \LL^{\dag} \bb(f) \ww_f^{1/2}\right| \le \rho,
\end{align*}
where the final bound follows from \cref{def:uncorrelated_edge_subset}.

To implement this algorithm in $\Otil(m)$ time,
take the aggregated vector
\[
\bb\left(F\right)
=
\sum_{e \in F} \bb\left( e \right) \ww_{e}^{1/2},
\]
and solve for $\xx(F) = \LL^{\dag} \bb(F)$ in $\Otil(m)$ time.
Then for each $f \in F$, computing $\bb(f)^{\top} \xx(F)$ involves
reading off two entries and taking their differences.
The assumption of $\rho > \frac{1}{\poly(n)}$ means it suffices to
compute all these to $\frac{1}{\poly(n)}$ additive error,
so the running time follows from guarantees of Laplacian solver
in \cref{lem:laplacian_solver}.
\end{proof}

\subsection{Estimating Linear Combinations of Leverage Scores}

Combining sketching with the approach of \cref{lem:estimate_first} gives an algorithm that produces unbiased estimates of linear combinations of leverage scores with low variance.

\begin{lemma}
\label{lem:estimate_second}
Given a $\rho$-correlated subset $F$
and coefficients $\ttheta \in [0, 1]^{F}$,
we can produce in $\Otil(m)$ time an unbiased estimator
to $\sum_{f \in F} \ttheta_f \ttau_f$ with variance at most $2\rho^2 |F|$.
\end{lemma}

\begin{proof}
Generate independent Radamacher random variables $\rr_f = \pm 1$ with
equal probability over each edge $f \in F$,
then compute the vector
\[
\vv
\leftarrow
\sum_{f \in F}
\rr_f 
\ww_f^{1/2} \ttheta_f^{1/2} \bb\left( f \right)
\]
and return $\vv^{\top} \LL^{\dag} \vv$.
The expression expands into
\begin{multline*}
\left( \sum_{f \in F}
\rr_f 
\ww_f^{1/2} \ttheta_f^{1/2} \bb\left( f \right) \right)
\LL^{\dag}
\left( \sum_{f \in F}
\rr_f 
\ww_f^{1/2} \ttheta_f^{1/2} \bb\left( f \right) \right)\\
=
\sum_{f \in F} \ttheta_f \ww_f
\bb(f)^{\top} \LL^{\dag} \bb(f)
+
\sum_{e \in F, f \in F, e \neq f} \rr_e \rr_f
\cdot \ttheta_e^{1/2}\ttheta_f^{1/2}
\cdot \ww_{e}^{1/2} \bb\left( e \right)^{\top} \LL^{\dag}
\bb\left( f \right) \ww_f^{1/2}.
\end{multline*}
The first term is precisely the value that we are estimating,
so it suffices to bound the expectation and variance of the trailing term.

The independence of $\rr_e$ and $\rr_f$ means
$\E_{\rr_e, \rr_f} \rr_e \rr_f
= (\E_{\rr_e} \rr_e) \cdot (\E_{\rr_f} \rr_f) = 0$ if $e \neq f$.
So it remains to calculate the expected value of its second moment.
\begin{align*}
&\E_{\rr}
\left[
\left(\sum_{e \in F, f \in F, e \neq f} \rr_e \rr_f
\cdot \ttheta_e^{1/2}\ttheta_f^{1/2}
\cdot \ww_{e}^{1/2} \bb\left( e \right)^{\top} \LL^{\dag}
\bb\left( f \right) \ww_f^{1/2}\right)^2
\right]
\\
= ~& \sum_{\substack{e_1, e_2, f_1, f_2 \in F \\ e_1 \neq f_1, e_2 \neq f_2}}
\E_{\rr}
\left[\rr_{e_1} \rr_{e_2} \rr_{f_1} \rr_{f_2} \right]
\cdot
\ttheta_{e_1}^{1/2}\ttheta_{f_1} ^{1/2}
\ww_{e_1}^{1/2} \bb\left( e_1 \right)^{\top} \LL^{\dag}
\bb\left( f_1 \right) \ww_{f_1}^{1/2}
\cdot
\ttheta_{e_2}^{1/2}\ttheta_{f_2}^{1/2}
\ww_{e_2}^{1/2} \bb\left( e_2 \right)^{\top} \LL^{\dag}
\bb\left( f_2 \right) \ww_{f_2}^{1/2}\\
\\ = ~&
2\sum_{e \in F, f \in F, e \neq f}
\ttheta_e \ttheta_f
\left( \ww_{e}^{1/2} \bb\left( e \right)^{\top} \LL^{\dag}
\bb\left( f \right) \ww_f^{1/2} \right)^2
\leq
2\sum_{e \in F, f \in F, e \neq f}
\left( \ww_{e}^{1/2} \bb\left( e \right)^{\top} \LL^{\dag}
\bb\left( f \right) \ww_f^{1/2} \right)^2.
\end{align*}
Here the second equality follows because the expectation equals $0$ unless $e_2 = e_1, f_2 = f_1$ or $f_2 = e_1, e_2 = f_1$,
and the last inequality follows from $\ttheta_e \leq 1$.
Using that the set $F$ of edges is $\rho$-correlated, we get
\[
\sum_{e \in F, f \in F, e \neq f}
\left( \ww_{e}^{1/2} \bb\left( e \right)^{\top} \LL^{\dag}
\bb\left( f \right) \ww_f^{1/2} \right)^2
\leq
\sum_{f \in F}
\left( \sum_{e} \left| \ww_{e}^{1/2} \bb\left( e \right)^{\top} \LL^{\dag}
\bb\left( f \right) \ww_f^{1/2} \right| \right)^2
=
2\left| F \right| \rho^2.
\]
\end{proof}

\subsection{Identifying Uncorrelated Edge Subsets}
\label{subsec:find_uncorrelated_subsets}

Uncorrelated edge subsets are found by sampling a random subset of edges.
It relies on the fact that the total correlation between
all edges is small. This is known as localization of electrical flows.

\begin{theorem}[Localization of electrical flows~{\cite[Theorem 1.5]{SRS18}}]
\label{thm:localization}
For any undirected weighted graph $G = (V, E, \ww)$, and vector $c \in \R^m$, it holds that
\[ \sum_{e, f \in E}
\left|\cc_e\ww_e^{1/2} \bb\left( e \right)^\top \LL^{\dag} \bb\left( f \right)  \ww_f^{1/2}\cc_f\right|
\leq
O\left(\|\cc\|_2^2 \log^2 n\right).
\]
\end{theorem}

Let us explain what the term ``localization'' means. In an unweighted graph, $|\bb(e)^{\top} \LL^{\dag} \bb(f)|$
represents the amount of flow on edge $e$ in the electrical flow
routing $1$ unit of demand between the two endpoints of $f$.
So this total summed over $e$ and $f$ being small can be interpreted
as saying that on average, each electrical flow stays close to
its starting/ending locations.

Our main procedure for finding uncorrelated edge subsets is below.
We state it as finding a subset of some initial set $S$
because \cref{lem:taylor} also needs the diagonal terms of $\MM$ to be large,
i.e., $\ttau_e$ is bounded away from $1$.
So we first restrict $S$ to these edges before taking a further subset:
in graphs that are not too sparse (which we will guarantee by removing low-degree vertices in \cref{lem:partial_exact_elimination}),
invoking standard leverage score estimation algorithms,
and keeping only edges with estimates bounded away from $1$
can guarantee $|S| \geq \Omega(m)$.

\begin{lemma}
\label{lem:get_uncorrelated}
There is an algorithm \textsc{GetUncorrelated}($G = (V, E, \ww)$, $S$, $k$) that given a weighted graph $G$, a subset
of edges $S \subseteq E$, and integer $k \le |S|/2$,
in time $\Otil(m)$ outputs a $O(\frac{k}{|S|} \log^2{n})$-correlated
subset $F \subseteq S$ of size $k$ with high probability.
\end{lemma}

Running this algorithm requires estimating the total correlation
between a subset of edges.
As in Schild's algorithms~\cite{LS18,S18}
we use of $\ell_1$-sketches for this:
the total correlation for an edge $f$ is the $\ell_1$ norm of
$\WW^{1/2} \BB^{\top} \LL^{\dag} \BB \WW^{1/2} \one_f$, restricted to coordinates excluding $f$ itself. We estimate this with $\ell_1$-stable, or Cauchy, distributions, plus a random sampling trick.

\begin{theorem}[\!\!{\cite[Theorem 3]{Indyk06}} stated for $\ell_1$]
\label{thm:l1}
Given an integer $d \ge 1, \delta \in (0, 1), \eps \in (0, 1)$ there is a random matrix $\CC \in
\R^{t \times d}$ with $t = O(\eps^{-2}\log(1/\delta))$ and an algorithm
$\textsc{Recover}(\uu, d, \delta, \eps)$ satisfying the following properties.
\begin{itemize}
\item For any vector $\vv \in \R^d$ the algorithm
$\textsc{Recover}(\CC \vv, d, \delta, \eps)$
with probability $1-\delta$ and runtime $\O(d)$ outputs an estimate $r$ with
\[
\left(1-\eps\right) \|\vv\|_1
\le
r
\le
\left(1+\eps\right) \|\vv\|_1.
\]
\item Every element of $\CC$ is independently sampled from a Cauchy distribution.
\end{itemize}
\end{theorem}

\begin{proof}[Proof of \cref{lem:get_uncorrelated}]
Let $F^+$ be a random subset of $S$ of size $2k$.
Let
\begin{align*}
    \MM :=     \WW_{F^{+}, F^{+}}^{1/2}
    \BB_{F^{+}, :} \LL^{\dag} \BB_{F^{+}, :}^{\top}
    \WW_{F^{+}, F^{+}}^{1/2}
    \in \R^{F^+ \times F^+}
\end{align*}
be the matrix whose $(i,j)$-th entry is the correlation between edge $i$ and $j$.
For each edge $f$ in $F^+$, we want to estimate the total correlation of it against all other edges in $F^+$, or equivalently, the $\ell_1$ norm of the column of $(\MM - \diag(\MM))$ corresponding to $f$. Formally, we have for any $f \in \F^+$,
\begin{align*}
    \| (\MM - \diag(\MM)) \one_{f} \|_1 = 
    \sum_{e \in F^+, e \neq f} \left|
        \ww_e^{1/2} \bb(e)^\top \LL^\dagger \bb(f) \ww_f^{1/2}
    \right|.
\end{align*}

Let $\calD$ denote the uniform distribution over the vectors in $\{0, 1\}^{F^+}$. For $\hh \in \{0,1\}^{F^+}$, we use $\bar{\hh}$ to denote the bitwise complement of $\hh$, by replacing each $0$ with $1$ and each $1$ with $0$.
Then, the matrix $4 \diag(\hh) \MM
    \diag(\bar{\hh})$ is an unbiased estimator for $\MM$ excluding the diagonal
\begin{align*}
    \E_{ \hh \sim \calD} \left[
    4 \diag(\hh) \MM
    \diag(\bar{\hh})
    \right]
    =
    \MM
    -
    \diag(
    \MM
    )
\end{align*}
since each off-diagonal entry is sampled with probability $1/4$. Independently sample $\hh_1, \hh_2, \dots, \hh_{\ell} \sim \calD$ for $\ell = O(\log n)$ and let $\widetilde{\MM}$ denote the average 
\begin{align*}
    \widetilde{\MM} := \frac{1}{\ell} \sum_{i=1}^{\ell} 
    4 \diag(\hh_i) \MM
    \diag(\bar{\hh_i}).
\end{align*}
Now, each off-diagonal entry of $\MM$ has weight contributing to $\widetilde{\MM}$ equal to the average of $\ell$ independent $\mathrm{Ber}(1/4)$ random variables. By Chernoff bounds, with probability at least $1 - n^{-3}$, the average is in range $[1/8, 3/8]$. Taking a union bound over the entries gives that $\widetilde{\MM}$ is an entry-wise $2$-multiplicative approximation of $\MM - \diag(\MM)$ with high probability, i.e., for all $e \neq f \in F$,
\begin{align*}
    \Pr\left[
        \frac{1}{2} |\widetilde{\MM}_{ef}| \le |\MM_{ef}| \le 2 |\widetilde{\MM}_{ef}|
    \right] > 1 - 1/n^3.
\end{align*}
Instead of computing $\widetilde{\MM}$ explicitly, we invoke \cref{thm:l1} with $\eps = 0.1$ and $d = |F^+|$,
and compute the matrix
\[
    \CC \widetilde{\MM}
    =
    \frac{1}{\ell} \sum_{i=1}^{\ell} 
    4 \CC \diag(\hh_i) \MM
    \diag(\bar{\hh_i})
\]
where $\CC$ is the random matrix given in \cref{thm:l1}. Now calling $\textsc{Recover}(\CC \widetilde{\MM}, |F^+|, 1/n^3, 0.1)$ gives an $1.1$-approximation of the $\ell_1$ norms of the columns of $\widetilde{\MM}$. Thus for each edge in $F^+$, we obtained an $3$-approximation of the total correlation against other edges in $F^+$. 
We keep the edge if it's less than $O( \frac{k}{|S|} \log^2 n)$.

Now we show that with constant probability there are at least $k$ edges left in $F^+$. 
For any $f \in E$ and subset $T \subseteq E$, let us denote
\begin{align*}
    \mathrm{Cor}(f, T) := \sum_{e \in T, e \neq f} \left|
        \ww_e^{1/2} \bb(e)^\top \LL^\dagger \bb(f) \ww_f^{1/2}
    \right|
\end{align*}
as the total correlation of the edge $f$ against the edges in $T \setminus \{ f \}$.
By \Cref{thm:localization} with $\cc = \one_S$, we bound the total correlation between the edges in $S$ 
\begin{align*}
    \sum_{f \in S} \mathrm{Cor}(f, S) \le O(|S| \log^2 n).
\end{align*}
Therefore, we have
\begin{align*}
    \E_{F^+, f \sim F^+} [\mathrm{Cor}(f, F^+)]
    =
    \frac{|F^+| - 1}{|S| - 1} \frac{1}{|S|} \sum_{f \in S} \mathrm{Cor}(f, S)
    \le 
    \frac{2k}{|S|^2} O(|S| \log^2 n) = O(k \log^2 n / |S|).
\end{align*}
By Markov's inequality, 
\begin{align*}
    \Pr_{F^+} \left[ 
        \E_{f \sim F^+} [\mathrm{Cor}(f, F^+)]
        \le O(k \log^2 n / |S|)
    \right] > 0.9.
\end{align*}
Assuming this holds, another Markov's inequality gives that 
\begin{align*}
    \Pr_{f \sim F^+} [\mathrm{Cor}(f, F^+)
        \le O(k \log^2 n / |S|) ] > 0.9,
\end{align*}
so at most $0.1|F^+|$ edges in $F^+$ are discarded in the last step, implying that at least $1.9k$ edges are kept. 
Taking what's left and further reducing
it to $k$ edges gives the desired $F$, since taking a subset can only decrease
correlation because correlations are all positive.

For the running time, in expectation we need to sample $O(1)$ many $F^+$'s to construct $F$. For each sample, computing the matrix $\CC \widetilde{\MM}$ requires summing $O(\log n)$ following matrices 
\begin{align*}
    \CC \cdot \diag(\hh_i) \MM \diag(\bar{\hh_i})
    & =
    \CC \cdot  \diag(\hh_i)
    \WW_{F^{+}, F^{+}}^{1/2}
    \BB_{F^{+}, :} \LL^{\dag} \BB_{F^{+}, :}^{\top}
    \WW_{F^{+}, F^{+}}^{1/2} \diag(\bar{\hh_i}).
    \\ & = 
    \CC \cdot \DD_1  
     \BB_{F^{+}, :} \LL^{\dag} \BB_{F^{+}, :}^{\top}
     \DD_2
\end{align*}
for some diagonal matrices $\DD_1$ and $\DD_2$.
By taking transpose, the above matrix becomes    
\begin{align*}
(\DD_2 \BB_{F^{+}, :})
    ( \LL^{\dag}
    (\CC \DD_1
    \BB_{F^{+}, :})^\top ).
\end{align*}
Since $\CC$ is a matrix with size $t \times (2k)$ for $t = O(\log n)$, we can solve $t$ Laplacian systems to compute the term  $( \LL^{\dag}
    (\CC \DD_1
    \BB_{F^{+}, :})^\top ) \in \R^{V \times [t]}$. Left multiplying by $\DD_2 \BB_{F^{+}, :} \in \R^{E \times V} $ can be done in $O(mt)$ since the matrix only has two nonzero entries per row. This also dominates the time to do the $\ell_1$-norm recovery.
\end{proof}

\section{Algorithm and Analysis}
\label{sec:ball}

We now give our overall algorithm.
To invoke the Taylor expansion-based approximation
from \cref{lem:taylor}, we also need our uncorrelated
edge subset to have leverage scores bounded away from $1$.
This is done by repeatedly eliminating away low degree
vertices.

\begin{lemma}
\label{lem:partial_exact_elimination}
Given a graph $G$ with $m$ edges,
we can obtain in $O(m)$ time a graph $H$ with at most $m$ edges,
minimum combinatorial vertex degree at least $3$,
along with an additive difference $\Delta$ such that
$\log \cT(G) = \log \cT(H) + \Delta$.
\end{lemma}

\begin{proof}
For a vertex of degree one with adjacent edge weight $w$, we can simply remove it and add $\log w$ to $\Delta$.

For a vertex $v$ of degree two, let $e_1, e_2$ denote its two adjacent edges with edge weights $w_1, w_2$ respectively. Every spanning tree of $G$ either contains both $e_1, e_2$ or contains one of $e_1, e_2$. We have
\begin{align*}
    \cT(G) = w_1 w_2 \cdot \cT_2 + (w_1 + w_2) \cdot \cT_1
\end{align*}
where we let $\hat{G} := (V \setminus \{v\}, E \setminus \{e_1, e_2 \}, \ww)$, $\cT_1 := \cT(\hat{G})$, and $\cT_2$ denote the total spanning two-forest weight of the graph $\hat{G}$. Here, a subgraph is a spanning two-forest if it can form a spanning tree after adding one edge. If we merge the two edges by one edge with weight $w$ (so $v$ is removed), the total spanning tree weight of the new graph $G_{\mathrm{new}}$ is 
\begin{align*}
    \cT(G_{\mathrm{new}}) = w \cT_2 + \cT_1.
\end{align*}
By setting $w = \frac{w_1 w_2}{ w_1 + w_2}$, we can relate $\cT(G)$ to $\cT(G_{\mathrm{new}})$ by
\begin{align*}
    \cT(G) = \cT(G_{\mathrm{new}}) \cdot (w_1 + w_2).
\end{align*}
Therefore, in this case we merge the two edges $e_1, e_2$ into one edge with weight $w$, and add $\log(w_1 + w_2)$ to $\Delta$.

Repeating the above process until no vertex has degree $\le 2$, we obtain the desired $H$ with minimal degree $\ge 3$ and a difference $\Delta$. This can be implemented in linear time, since each operation can be done in $O(1)$ time, and it suffices to check the neighborhood after removing a vertex.

\end{proof}

After applying \cref{lem:partial_exact_elimination} to a graph $G$, the number of edges $m$ will be at least $\frac32n$, and thus the average leverage score is at most $2/3$.
We will identify edges with these low leverage scores via $\ell_2$/Johnson-Lindenstrauss sketches, and
set the initial set $S$ in \cref{lem:get_uncorrelated} to such edges.

\begin{theorem}[Leverage score estimation \cite{SS11}]
\label{thm:l2}
There is an algorithm that given an undirected graph $G = (V, E, \ww)$
with $m$ edges, poly-bounded weights,
and parameter $\eps \in (0, 1)$,
outputs estimates $\ttautil_e$ satisfying
\[ (1-\eps)\ttautil_e \le \ww_e \bb(e)^\top \LL^{\dag} \bb(e) \le (1+\eps)\ttautil_e. \]
The algorithm runs in time $\O(m\eps^{-2})$ and succeeds with high probability.
\end{theorem}

We now have all the tools we need to state the main algorithm. It is very simple: first delete low-degree vertices using \cref{lem:partial_exact_elimination}, then find an uncorrelated edge subset $F \subseteq E(G)$ of size about $\sqrt{m}$, and then estimate $\cT(G)/\cT(G \setminus F)$ using properties of $F$.
\begin{figure}[!ht]
    \begin{algbox}
        $\textsc{ApproxSpanningTree}(\hat{G}, \eps)$: Given an undirected graph $\hat{G}$, error parameter $\eps$, w.h.p. returns an $O(\eps)$-additive approximation of the logarithm of the spanning tree count.
        
        \begin{enumerate}
            \item Eliminate vertices with degree $\le 2$ according to \Cref{lem:partial_exact_elimination}. Let $G = (V, E, \ww)$ denote the resulting graph, and store the additive difference.
            \item Obtain estimates of the leverage scores for all edges of $G$ by \Cref{thm:l2} with $\eps = 0.1$. Let $S \subseteq E(G)$ be the set of edges whose estimate is at most $0.8$.
            \item $F \gets \textsc{GetUncorrelated}(G, S, k)$ (see \cref{lem:get_uncorrelated}) for $k = \Theta( \eps \sqrt{m} / (\log m)^3)$, where $m$ is the number of edges in $G$.
            \item Apply \Cref{lem:estimate_first} to get leverage score estimates $\ttautil_f$ for all $f \in F$.
            \item Apply \Cref{lem:estimate_second} with $\ttheta_f = C \cdot \frac{1}{1 - \ttautil_f}$ where $C = 1/10$. Let $\phi$ denote the estimate of $\sum_{f \in F} \ttheta_f \ttau_f$ divided by $C$.
            \item \Return $\textsc{ApproxSpanningTree}((V, E \setminus F, \ww), \eps) - \left( - \phi+  \sum_{f \in F} \log(1 - \ttautil_f) + \frac{\ttautil_f}{1 - \ttautil_f}  \right).$ 
        \end{enumerate}
    \end{algbox}
    \caption{Our recursive algorithm for undirected graph spanning tree approximation}
\label{fig:alg1}
\end{figure}

Towards analyzing the algorithm in \cref{fig:alg1}, we prove that the estimate provided in each iteration has small additive error and variance compared to the true change in spanning tree count between the graphs $G$ and $G \setminus F$.

\begin{lemma}
\label{lemma:helper}
Let $G = (V, E, \ww)$ be a graph with minimum degree $\ge 3$. Let $F \subseteq E(G)$ be a subset of edges such that (1) all $e \in F$ have leverage score at most $0.88$, i.e., $\ww_e \bb(e)^\top \LL^\dagger \bb(e) \le 0.88$, and (2) $F$ is a $\rho$-correlated subset for $\rho \le 0.01$. Consider running steps $4$, $5$ in the algorithm of \cref{fig:alg1} on $G$ and $F$. Define the corresponding random variable
\[ X := - \phi+  \sum_{f \in F} \log(1 - \ttautil_f) + \frac{\ttautil_f}{1 - \ttautil_f}. \]
Then we have that
\[ \Big|\E[X] - \E[\log \cT(G \setminus F) - \log \cT(G)] \Big| \le O(|F|\rho^2) \enspace \text{ and } \enspace \Var[X] \le O(|F| \rho^2). \]
\end{lemma}
\begin{proof}
By the calculation at the start of \cref{sec:algorithm} and \cref{lem:taylor} we know that
\[ \left|\log \cT(G \setminus F) - \log \cT(G) - \sum_{e \in F} \log(1 - \ttau_e) \right| \le O(|F|\rho^2) \]
where $\ttau_e = \ww_e \bb(e)^\top \LL^\dagger \bb(e)$. Because $|\ttau_e - \ttautil_e| \le 2\rho$ by \cref{lem:estimate_first}, we know that $\ttautil_e \le 0.9$ and
\[ \log(1 - \ttau_e) - \log(1 - \ttautil_e) = \log\left(1 - \frac{\ttau_e - \ttautil_e}{1 - \ttautil_e}\right) = -\frac{\ttau_e}{1 - \ttautil_e} + \frac{\ttautil_e}{1 - \ttautil_e} + O(\rho^2), \]
for all $e$. Finally, by \cref{lem:estimate_second} we know that $\E[\phi] = \sum_{e \in F} \frac{\ttau_e}{1 - \ttautil_e}$ and $\Var[X] = \Var[\phi] \le O(|F|\rho^2)$. Combining these estimates completes the proof.
\end{proof}

\begin{proof}[Proof of \Cref{thm:main}]
Let us first establish that during the algorithm, hypotheses (1) and (2) of \cref{lemma:helper} hold during all steps with high probability, for $\rho = O(\frac{k}{m}(\log m)^2)$ where $m$ is the number of edges at that iteration. (1) follows because $S$ only contains edges whose leverage score estimate is at most $0.8$ (and the estimate is accurate up to $0.1$ multiplicative error), and because \cref{thm:l2} succeeds with high probability.
For (2), we first establish that $|S| \ge m/21$. Indeed, $S$ will contain all edges whose leverage scores are at most $0.7$. By Markov's inequality, the number of such edges is at least $m - \frac{n}{0.7} \ge m - \frac{10}{7} \cdot \frac{2}{3} m = m/21$. Thus, step 3 produces a set $F$ that is $O(\frac{k}{m} (\log m)^2)$-correlated.

Now let us bound the additive error and variance of the estimate output. Split the algorithm into phases, depending on how many remaining edges the graph has. Phase $i$ starts when the graph has at most $m_i := 2^{-i}m$ edges, for $i \le O(\log m)$. Let $k_i = \Theta(\eps\sqrt{m_i}/(\log m_i)^3)$, so that phase $i$ lasts for at most $m_i/k_i$ iterations. By \cref{lemma:helper}, the total additive error and variance in the $i$-the phase is bounded by
\[ O\left(\frac{m_i}{k_i} \cdot k_i \cdot \Big(\frac{k_i}{m_i} (\log m_i)^2 \Big)^2 \right) \le O\left(\frac{k_i^2}{m_i} (\log m_i)^4 \right) \le O(\eps^2 (\log m_i)^{-2}). \]
Here, we have used that variance over independent events is additive.
Over all phases, this sums to $O(\eps^2)$, so with constant probability, the estimate has error $O(\eps)$.

The runtime follows because there are $\O(m^{1/2}\eps^{-1})$ iterations, each of which takes $\O(m)$ time. One can run the algorithm $O(\log m)$ times and take the median output to get a high-probability success rate.
\end{proof}

\section*{Acknowledgements}

This project is motivated by the Altius restaurant in Pittsburgh.

\bibliographystyle{alpha}
\bibliography{ref}

\end{document}